\documentclass[11pt]{article}

\usepackage{fullpage}

\usepackage{dsfont}
\usepackage{latexsym}
\usepackage{natbib}
\usepackage{amsmath}
\usepackage{mathtools}
\usepackage{amsthm}
\usepackage{amssymb}
\usepackage{amsfonts}
\usepackage{mathrsfs}
\usepackage{aliascnt}
\usepackage{graphicx}
\usepackage[dvipsnames]{xcolor}
\usepackage{color}
\usepackage{enumerate}
\usepackage{multirow}

\usepackage[pdfpagelabels,pdfpagemode=None]{hyperref}
\usepackage{algorithmic}
\usepackage{algorithm}

\newcommand{\third}{\frac{1}{3}}
\newcommand{\twothirds}{\frac{2}{3}}
\newcommand{\fourthirds}{\frac{4}{3}}

\newtheorem{theorem}{Theorem}%

\newaliascnt{lemma}{theorem}
\newtheorem{lemma}[lemma]{Lemma}%
\aliascntresetthe{lemma}

\newaliascnt{claim}{theorem}
\aliascntresetthe{claim}

\newaliascnt{corollary}{theorem}
\aliascntresetthe{corollary}

\newaliascnt{proposition}{theorem}
\aliascntresetthe{proposition}

\newaliascnt{remark}{theorem}

\aliascntresetthe{remark}

\newaliascnt{result}{theorem}
\newtheorem{result}[result]{Result}
\aliascntresetthe{result}

\newaliascnt{algo}{procedure}
\aliascntresetthe{algo}

\theoremstyle{definition}
\newtheorem{definition}{Definition}

\newtheorem{example}{Example}

\newcommand{\AutoAdjust}[3]{\mathchoice{ \left #1 #2  \right #3}{#1 #2 #3}{#1 #2 #3}{#1 #2 #3} }
\newcommand{\Xcomment}[1]{{}}

\newcommand{\InBrackets}[1]{\AutoAdjust{[}{#1}{]}}
\newcommand{\Ex}[2][]{\operatorname{\mathbf E}_{#1}\InBrackets{#2}}
\newcommand{\Exlong}[2][]{\operatornamewithlimits{\mathbf E}\limits_{#1}\InBrackets{#2}}

\newcommand{\vect}[1]{\ensuremath{\mathbf{#1}}}

\newcommand{\R}{\mathbb{R}}

\newcommand{\nbidder}{n}

\newcommand{\nitem}{m}

\newcommand{\util}{u}
\newcommand{\bundle}{S}
\newcommand{\win}{W}
\newcommand{\wini}[1][i]{\win_{#1}}
\newcommand{\opt}{\mathrm{OPT}}

\newcommand{\val}{v}
\newcommand{\vali}[1][i]{\val_{#1}}
\newcommand{\vals}{\vect{\val}}
\newcommand{\valsmi}[1][i]{\vals_{-#1}}
\newcommand{\hval}{\widehat{v}}
\newcommand{\hvali}[1][i]{\hval_{#1}}

\newcommand{\price}{p}
\newcommand{\prices}{\vect{\price}}
\newcommand{\pricei}[1][i]{{\price_{#1}}}
\newcommand{\vpricei}[1][i]{{\vec {\price}_{#1}}}
\newcommand{\pricew}{\widetilde{p}}

\newcommand{\vprice}{\vec{\price}}
\newcommand{\qprice}{q}

\newcommand{\qpricew}{\widetilde{q}}

\newcommand{\bid}{b}
\newcommand{\bidi}[1][i]{{\bid_{#1}}}
\newcommand{\bids}{\vect{\bid}}
\newcommand{\bidsmi}[1][i]{{\bids_{-#1}}}

\newcommand{\vbidi}[1][i]{{\vec{\bid}_{#1}}}

\newcommand{\wbid}{\widetilde{b}}
\newcommand{\wbidi}[1][i]{{\wbid_{#1}}}

\newcommand{\projbid}{\varphi}

\newcommand{\alloc}{X}
\newcommand{\alloci}[1][i]{{\alloc_{#1}}}
\newcommand{\allocs}{\vect{X}}

\newcommand{\strat}{s}
\newcommand{\strati}[1][i]{{\strat_{#1}}}
\newcommand{\strats}{\vect{\strat}}
\newcommand{\stratsmi}[1][i]{{\strats_{-#1}}}

\newcommand{\wstrat}{\widetilde{s}}
\newcommand{\wstrati}[1][i]{{\wstrat_{#1}}}

\newcommand{\dist}{\mathcal{D}}
\newcommand{\distw}{\widetilde{\mathcal{D}}}
\newcommand{\dists}{\vect{\dist}}

\newcommand{\valdist}{\mathcal{F}}
\newcommand{\valdists}{\vect{\valdist}}

\begin{document}

\title{Simultaneous Auctions are (almost) Efficient}

\author{
Michal Feldman\thanks{Hebrew University of Jerusalem, and Harvard University, \texttt{michal.feldman@huji.ac.il}} 
\and
Hu Fu\thanks{Dept.\@ of Computer Science, Cornell University, \texttt{hufu@cs.cornell.edu}}
\and
Nick Gravin\thanks{Nanyang Technological University, \texttt{ngravin@pmail.ntu.edu.sg}} 
\and
Brendan Lucier\thanks{Microsoft Research New England, \texttt{brlucier@microsoft.com}
\newline Work was done while first three authors were visiting Microsoft Research New England.}
}

\maketitle

\begin{abstract}
Simultaneous item auctions are simple procedures 
for allocating items to bidders with 
potentially complex preferences over different item sets. 
In a simultaneous auction, every bidder submits bids on all items simultaneously. 
The allocation and prices are then resolved for each item separately, based solely on the bids submitted on that item.
Such procedures occur in practice (e.g.\ eBay) but are not truthful.
We study the efficiency of Bayesian Nash equilibrium (BNE) outcomes of simultaneous first- and second-price auctions when bidders have 
complement-free (a.k.a. subadditive) valuations. We show that the expected social welfare of any BNE is at least $\tfrac 1 2$ of the optimal social welfare 
in the case of first-price auctions, and at least $\tfrac 1 4$ in the case of second-price auctions.
These results improve upon the previously-known logarithmic bounds, which were established by \citet{HKMN11} for first-price auctions and by \citet{BR11} for second-price auctions.

\end{abstract}

\section{Introduction}
\label{sec:introduction}

The central problem in algorithmic mechanism design is to determine how best to allocate resources among individuals, while respecting both computational constraints and the individual incentives of the participants.  Much of the theoretical work in this field to date has focused on solving such problems \emph{truthfully}.
In a truthful mechanism, the participants reveal their preferences in full to a central orchestrator, who then
distributes the resources in a way
that incentivizes
truthful revelation.
Such an approach has theoretical appeal, but truthful mechanisms tend to be complex and are rarely used in practice.
Instead, it is common to forego truthfulness and use simpler mechanisms.
Canonical examples of such auctions are the generalized second price (GSP) auctions for online advertising \citep{Edelman05,Varian07}, and the ascending price auction for electromagnetic spectrum allocation \citep{Milgrom98}.
Given that such simple auctions are used in practice, it is of crucial importance to determine how they actually perform when used by rational (and strategic) agents.

Consider the problem of resolving a \emph{combinatorial auction}.  In such a problem there is a large set $M$ of $m$ objects for sale, and $n$ potential buyers.  Each buyer has a private value function $\vali \colon 2^M \to \mathbb{R}_{\geq 0}$ mapping sets of objects to their associated values.  The goal of the market designer is to decide how to allocate the objects among the buyers to maximize the overall social efficiency.  One approach would be to elicit the valuation function from each bidder, then attempt to solve the resulting optimization problem.  However, in existing online marketplaces such as eBay, buyers do not express their (potentially complex) preferences directly; rather, each item is auctioned independently, and a buyer is forced to bid separately on individual items.  This approach is simple and natural, and relieves the burden of expressing a potentially complex valuation function.  On the other hand, this limited expressiveness could potentially lead to inefficient outcomes. This begs the question: how well does the outcome of simultaneous item auctions approximate the socially optimal allocation?

In order to evaluate the performance of non-truthful mechanisms, we take the economic viewpoint that self-interested agents will apply bidding strategies at equilibrium, so that no agent can unilaterally improve his outcome by changing his strategy.  We apply a quantitative approach, and ask how well the performance at equilibrium approximates the socially optimal outcome.  Since there may potentially be multiple equilibria, we will bound the performance in the worst case over equilibria.  Put another way, our approach is to use the {\em price of anarchy} as a performance measure for the analysis of mechanisms.

The fact that equilibria of simultaneous auctions might not be socially optimal was first observed by \citet{Bikhchandani99}, who studied the complete information\footnote{In a complete (or full) information setting, it is assumed that the bidders' valuations are commonly known to all participants} setting.  As he states:

\begin{quote}
``Simultaneous sealed bid auctions are likely to be inefficient under complete information
and hence, also under the more realistic assumption of incomplete information
about buyer reservation values."
\end{quote}

Our goal is to bound the extent of this inefficiency in the incomplete information setting.  To this end,
we model incomplete information using the standard Bayesian framework.
In this model, the buyers' valuations are assumed to be drawn independently from (not necessarily identical) distributions.  This product distribution is commonly known to all of the participants; we think of this as representing the public's aggregate beliefs about the buyers in the market.  While the distributions are common knowledge, each agent's true valuation is private.  This Bayesian model generalizes the full-information model of Nash equilibrium, which implicitly supposes that the type profile is known by all participants.  Note that while the agents are aware of the type distribution, the mechanism (which applies simultaneous item auctions) is prior-free and hence agnostic to this information.

\paragraph{Pricing and Efficiency in Simultaneous Auctions.}
We consider separately the case in which items are sold via first-price auctions (in which the player who bids highest wins and pays his bid), and the case of second-price auctions\footnote{Second-price item auctions are also known as Vickrey auctions; we will use these terms interchangeably.} (in which the winning bidder pays the second-highest bid).  The differences between first and second-price item auctions have received significant attention in the recent literature.  For example, a pure Nash equilibrium of our mechanism with simultaneous first-price auctions is equivalent to a Walrasian equilibrium \citep{Bikhchandani99,HKMN11}, and therefore must obtain the optimal social welfare \citep{BM97}.  On the other hand, every pure Nash equilibrium for second-price auctions is equivalent to a Conditional equilibrium, and hence obtains at least half of the optimal social welfare \citep{FKL12}.  While these constant factor bounds are appealing, their power is marred by the fact that pure equilibria do not exist in general.

Can we hope for such constant-factor bounds to hold for general Bayes-Nash equilibria?  For general valuations the answer is no.  Consider, for example, the case of a buyer who has a very large value for the set of all objects for sale, but no value for any strict subset.  In this case, any positive bid carries great risk: the buyer might win some items but not others, leaving him with negative utility.  It therefore seems that complements do not synergize well with item bidding, and indeed it has been shown by \citet{HKMN11} that the price of anarchy (with respect to mixed equilibria) in a first-price auction can be as high as $\Omega(\sqrt{m})$ when bidders' valuations exhibit complementarities.
The same lower bound can be easily extended to the case of second-price auctions\footnote{As explained in the sequel, to obtain meaningful results in second-price auctions one needs to impose {\em no-overbidding} assumptions on the bidding strategies, defined formally in Section \ref{sec:nob}. The $\Omega(\sqrt{m})$ lower bound extends to the case of second-price auctions under the {\em weak no-overbidding} assumption. The alternative \emph{strong no-overbidding} assumption is meaningless in the case of complements, as it precludes item bidding altogether.}.

Our main result is that \emph{the presence of complements is the only barrier to a constant price of anarchy}.  We show that when buyer valuations are complement-free (a.k.a.\ subadditive), the (Bayesian) price of anarchy of the simultaneous item auction mechanism is at most a constant, in both the first- and second-price auctions.

For first-price auctions, we show that any Bayes-Nash equilibrium yields at least half of the optimal social welfare.  This improves upon the previously best-known bound of $O(\log n)$ due to \citet{HKMN11}, where $n$ is the number of bidders.

\begin{result} [\bf BPoA$\leq 2$ in simultaneous first-price auctions.]
When buyers have subadditive valuations, the Bayesian price of anarchy of the simultaneous first-price item auction mechanism is at most $2$.
\end{result}

For simultaneous Vickrey auctions, it is not possible to bound the worst-case performance at equilibrium, even when there is only a single object for sale.  This impossibility is due to arguably unnatural equilibria in which certain players grossly overreport their values, prompting others to bid nothing.  To circumvent this issue one must impose an assumption that agents avoid such ``overbidding'' strategies.  In the \emph{strong no-overbidding assumption}, used by \citet{CKS08} and \citet{BR11}, it is assumed that each agent $i$ chooses bids so that, for every set of objects $S$, the sum of the bids on $S$ is at most $\vali(S)$.  
We show that under this assumption, the Bayesian price of anarchy for simultaneous Vickrey auctions is at most $4$.

\begin{result} [\bf BPoA$\leq 4$ in simultaneous second-price auctions.]
When buyers have subadditive valuations, the Bayesian price of anarchy of the simultaneous Vickrey auction mechanism is at most $4$, under the strong no-overbidding assumption.
\end{result}

The strong no-overbidding assumption is quite strong, as it must hold for {\em every} set of items.
A somewhat weaker assumption, referred to as {\em weak no-overbidding}, requires the the no overbidding condition holds only in expectation over the distribution of sets won by a player at equilibrium.
That is, agents are said to be \emph{weakly no-overbidding} if they apply strategies such that expected value of each agent's winnings is at least the expected sum of his winning bids \citep{FKL12}.
Roughly speaking, weak no-overbidding supposes that agents are generally averse to winning sets with bids that are higher than their true values.  However, unlike strong no-overbidding, it does not preclude strategies in which an agent overbids on sets that he does not expect to win, i.e.\ in order to more accurately express his willingness to pay for other sets.
For an expanded discussion of the no-overbidding assumptions, see \autoref{sec:no-overbidding-discussion}.

Notably, the BNE outcomes under the two no-overbidding assumptions are incomparable;
while the weak assumption is more permissive, and thus enables a richer set of behaviors in equilibrium,
it also introduces new ways to deviate from the prescribed equilibrium.
We show that the bound of $4$ on the Bayesian PoA extends also to the case of weakly no-overbidding agents.

\citet{BR11} showed that, under the strong no-overbidding assumption, the Bayeisan price of anarchy of the simultaneous Vickrey auction is strictly greater than $2$, and furthermore the price of anarchy is $\Omega(n^{1/4})$ when agent values are allowed to be correlated.  We show that similar results hold also under the weak no-overbidding assumption, proving bounds strictly greater than $2$ and $\Omega(n^{1/6})$, respectively.

Our bounds hold for subadditive bidders, whereas constant bounds on Bayesian price of anarchy were previously known only for the subclass of fractionally subadditive (i.e.\@ XOS) valuations \citep{CKS08}. Subadditive valuations are more expressive than their XOS counterparts, and obtaining price of anarchy bounds for subadditive valuations is significantly more challenging.  In particular, for XOS valuations, a player who aims to win certain set $S$ has a natural choice of bid: the additive valuation that determines his value for set $S$.  For subadditive valuations, there is no such notion of a natural bid aimed at representing one's value for a particular set, and hence even determining how best to bid on a certain set of interest is a non-trivial task.


%




\paragraph{Related Works}
Combinatorial auctions is a canonical subject of study in algorithmic mechanism design (see \citealp{NRTV07} and references therein for the large body of literature on this subject).
While most previous work focuses on the design of truthful mechanisms, we follow the more recent literature on the analysis of simple and practical (albeit not truthful) auctions.
Following the rich literature on the {\em price of anarchy} (PoA) \citep[see, e.g.,\@][for references]{RT07}, \citet{CKS08} pioneered the study of the \emph{Bayesian price of anarchy} (BPoA) and applied it to item-bidding auctions.
They bounded the BPoA by~$2$ in simultaneous second-price auctions with XOS valuations, which are equivalent to fractionally subadditive functions \citep{Feige09}.
The same bound was extended to the more general class of subadditive valuations by \citet{BR11}, and later to general valuations by \citet{FKL12}, albeit only with respect to \emph{pure} equilibria (when they exist).
The pure PoA was studies also in simultaneous first-price auctions by \citet{HKMN11}, who
showed a pure PoA of~$1$ for general valuations \footnote{Pure Nash equilibria rarely exist in
this case though, as they are shown to be equivalent to Walrasian equilibria of the corresponding two-sided market.}.

For both first- and second-price simultaneous auctions, the BPoA for subadditive valuations was not previously known to be better than $O(\log \nbidder)$. Previous techniques applied the known bounds for XOS valuations, using the $O(\log \nbidder)$ separation between XOS and subadditive valuations \citep[see e.g.\@][]{BR11}.

Studies on PoA and BPoA have provided insights into other settings, e.g.\ auctions employing greedy algorithms \citep{LB10}, Generalized Second Price Auctions \citep{PT10,LP11,CKKK11}, and also game-theoretic settings that are not related to auctions, such as network formation games \citep{Alon10}.

The {\em smoothness} technique for Bayesian games, developed by \citet{R12} and \citet{S12}, provides a method for extending bounds on pure PoA to Bayesian PoA.  However, to the best of our knowledge, our approach does not fall within this framework.  Roughly speaking, the smoothness framework requires that each player can find a good ``default'' strategy given his type, which is independent of the opponents' strategy selections.  However, subadditive valuations do not seem to admit such bids\footnote{We note that one can apply the technique on XOS
valuations, but because of the $O(\log \nbidder)$ separation between XOS and subadditive valuations \citep[see e.g.\@][]{BR11}, this gives only a logarithmic bound.}, and indeed the strategies we consider in our analysis depend heavily on the distribution of strategies applied by all players at equilibrium.



\paragraph{Organization of the paper.}
We introduce the necessary background and notation in \autoref{sec:prelim}. Our analysis then proceeds in two parts.  In the first part, \autoref{sec:subadditive}, we consider a single-player game in which the player, a subadditive buyer, must determine how best to bid on a set of objects against a distribution over price vectors.  We show that, for every distribution for which the expected sum of prices is not too large, the buyer has a bidding strategy that guarantees a high expected utility (compared to the player's value for the set of all objects).

In the second part of our analysis for the first-price (\autoref{sec:poa-fpa}) and Vickrey (\autoref{sec:poa-spa}) auctions, we show that every Bayes-Nash equilibrium must have high expected social welfare. We do this by considering deviations in which an agent uses the bidding strategy from the single-player game described in \autoref{sec:subadditive}, applied to some subset of the objects.  This subset of objects is chosen randomly: agent $i$ draws a new profile of types for his opponents from the type distribution, then considers bidding for the set he would be allocated under this ``virtual'' type profile.  At a BNE, agent $i$ cannot benefit from such a randomized deviation; we show this implies 
that the social welfare at equilibrium is at least a constant times the optimal welfare.

\section{Preliminaries}
\label{sec:prelim}

%
%


\subsection{Auctions and Equilibria}

\paragraph{Combinatorial Auctions.}
In a combinatorial auction, $\nitem$ items are sold to $\nbidder$ bidders.  Each bidder has a private combinatorial valuation
captured by a set function
$\val:2^{[\nitem]}\to\R_+$ over different bundles $\bundle\subseteq[\nitem]$.  Throughout the paper we assume the
valuations are \emph{monotone}, i.e.\ for every subset $\bundle \subseteq T \subseteq [\nitem]$ it holds that $\val(\bundle) \leq \val(T)$.
In a {\em Bayesian} (partial-information) setting, the bidders' valuation profile $\vals$ is drawn from a commonly known product
distribution\footnote{Whenever an expectation is taken with
respect to valuations, it will be assumed that they are drawn from these corresponding distributions.}
$\valdists = \valdist_1 \times \cdots \times \valdist_n$.
The outcome of an auction consists of an allocation $\allocs = (\alloc_1, \cdots, \alloc_\nbidder) \in 2^{[\nitem]
\times n}$, where $\alloc_i$ is the bundle of items allocated to bidder~$i$, and payments made by each bidder.  The
\emph{social welfare} of an allocation is $\sum_{i \in [\nbidder]} \vali(\alloci)$.
For any given valuation profile $\vals$, we let $(\opt_1^\vals, \dotsc, \opt_n^\vals)$ denote the welfare-maximizing assignment for profile $\vals$.

\paragraph{Simultaneous Item-Bidding Auctions.}
In a simultaneous item-bidding auction,
each bidder simultaneously submits a vector of bids, one for each item. The outcome of the auction is then decided item by item
according to the bids placed on each item.  In this paper we study two forms of such auctions: \emph{simultaneous first
price auctions} and \emph{simultaneous second price auctions}\footnote{The word ``simultaneous'' is often omitted, as we
study only simultaneous (in contrast to sequential) auctions.}.  In both auctions, each item is allocated to the bidder
who has placed the highest bid on it (breaking ties arbitrarily but consistently).  In a (simultaneous) first price
auction, the winner of each item pays his bid on that item, and in a (simultaneous) second price auction, the winner of
each item pays the second highest bid on that item.  We now give a more formal description of this process.

We generally write $\bidi(j)$ to denote the bid of player~$i$ on item~$j$, and
$\vbidi$ for the vector of bids placed by bidder $i$.  Alternatively, we may think of agent $i$'s bid $\bidi$ as an additive function
$\bidi(S)=\sum_{j\in S}\bidi(j)$ that corresponds\footnote{There is an easy equivalence between an additive function
$a(\bundle) \coloneqq \sum\limits_{j\in\bundle}a(\{j\})$ and its concise vector description
$\vec{a}=(a(\{1\}),\dots, a(\{\nitem\}))$. We will use functional and vector
representations interchangeably as the situation demands.
} to the bid-vector $\vbidi$.
Given a sequence of bid profiles $\bids=(\bid_1,\dots,\bid_\nbidder)$, we write $\wini(\bids)$
for the set of items won by bidder $i$, and $\vpricei \in \R_+^{\nitem}$ the vector of payments made by bidder~$i$ on the items.  In this notation, the
first- and second-price auctions can be summarized as follows:
\begin{equation*}
\begin{array}{c|c|c|}
\cline{2-3}
 & \text{First-price} & \text{Vickrey} \\
\cline{2-3}
\text{won set:} & \multicolumn{2}{c|}{\wini(\bids) = \{j \in [\nitem] \: | \: \bidi(j) > \bidi[k](j), \forall k \neq i\}}\\
\cline{2-3}
\text{payment:}
&
\pricei(j) = \left\{
\begin{array}{ll}
\bidi(j), & j \in \wini(\bids) \\
0, & j \notin \wini(\bids)
\end{array}
\right.
&
\pricei(j) = \left\{
\begin{array}{ll}
\max_{k \neq i} \bidi[k](j), & j \in \wini(\bids) \\
0, & j \notin \wini(\bids)
\end{array}
\right. \\
\cline{2-3}
\end{array}
\end{equation*}


%

We assume bidders have quasi-linear utilities, i.e.\
the \emph{utility} of bidder~$i$ for a given bid profile $\bids$ is
given by $\util_i(\bids) = \vali(\wini(\bids)) - \pricei(\wini(\bids))$.

\paragraph{A Single Bidder's Perspective on Bidding}
In both first and second price auctions, the set of items won by a bidder~$i$ bidding~$\bid_i$ is determined solely by a
coordinate-wise comparison between $\bid_i$ and the largest bid placed by the other bidders.  Let $\projbid_i(\bidsmi)$
be the vector whose $j$-th component is $\max_{k \neq i} \bidi[k](j)$.  It is often convenient to
write $\win(\bid_i, \bidsmi)$ as $\win(\bidi, \vprice)$ where $\vprice = \projbid_i(\bidsmi)$.
We think of $\vprice$ as the vector of prices perceived by bidder~$i$:
in the second price auction, the bidder pays the price on an item if his bid exceeds it; and in the first price auction the bidder pays his own bid on such an
item, and $\vprice$ is the minimum such winning bid.  It is in this light that we often write $\projbid_i(\bidsmi)$ as prices~$\vprice$ when this causes no confusion.  We will also shorten the notation $\val(\win(\bid, \vprice))$ to $\val(\bid,\vprice)$, meaning the value obtained when bidding $\bid$ against perceived prices $\vprice$.

\paragraph{Strategies and Equilibria.}
Buyers select their bids strategically in order to maximize utility.
The bidding behavior of a buyer given its valuation is described by a \emph{strategy}.  A strategy~$\strat_i$ maps each valuation $\val_i$ to a distribution over bid vectors; we interpret
$\strat_i(\val_i)$ as the (possibly randomized) set of bids placed by bidder $i$ when his type is $\val_i$.

\begin{definition}[Bayesian Nash Equilibria]
\label{def:bne}
A profile of strategies $\strats = (\strat_1(\val_1), \dotsc, \strat_\nbidder(\val_\nbidder))$ is in \emph{Bayes-Nash equilibrium} (BNE) for distribution $\valdists$ if, for every buyer $i$, type $\vali$, and bidding strategy $\wstrati$,
\[
   \Ex[\valsmi]{\Exlong[\substack{\bidsmi \sim \strats(\valsmi),\\ \bidi\sim\strati(\vali)}]{\util_i(\bidi,\bidsmi)}}
   \geq
   \Ex[\valsmi]{\Exlong[\substack{\bidsmi \sim \strats(\valsmi),\\ \wbidi\sim\wstrati}]{\util_i(\wbidi, \bidsmi)}}.
\]
\end{definition}

Given Fubini's Theorem, we can shorten the condition as follows (such shorthand forms are used throughout the paper):
\begin{equation}
\label{eq:BNEcond}
\Ex[\valsmi, \bids \sim \strats(\vals)]{\util_i(\bids)}  \geq  \Ex[\valsmi, \bids \sim \strats(\vals), \wbidi\sim\wstrati]{\util_i(\wbidi, \bidsmi)}.
\end{equation}

\begin{definition}[Bayesian Price of Anarchy]
\label{def:bpoa}
Given an auction type (either first- or second-price), the \emph{Bayesian price of anarchy} (BPoA) is
the worst-case ratio between the optimal expected welfare and the expected welfare at a BNE and is given by
\[
   \max_{\substack{(\valdists,\ \strats):\\ \strats \text{ a BNE for } \valdists}}
   \frac{\Ex[\vals]{\sum_i \val_i(\opt_i^\vals)}}{\Ex[\vals,\bids \sim \strats(\vals)]{\sum_i \val_i(\win_i(\bids))}}.
\]
\end{definition}
For second price auctions we will consider BPoA under natural restrictions on the strategies used by the bidders.  In
such cases, the maximum in \autoref{def:bpoa} is taken with respect to BNE under that restricted class of strategies.
We note that a BNE is guaranteed to exist as long as the space of valuations and potential bids is discretized, say with all values expressed as increments of some $\epsilon > 0$.
A more detailed discussion of BNE existence is given in Appendix \ref{app:existence}.

\subsection{Subadditive Valuations}

We focus on valuations that are complement-free in the following general sense:

\begin{definition}
\label{def:subadd}
A set function $\val: 2^{[\nitem]}\to\R_{+}$ is \emph{subadditive} if, for any
subsets $S_1,S_2\subset [\nitem]$,
\[
\val(S_1)+\val(S_2)\ge\val(S_1\cup S_2).
\]
\end{definition}

The class of subadditive functions strictly includes a hierarchy of more restrictive complement-free functions such as
submodular and gross substitute functions (see \citealp{LLN06} for definitions and discussions).  Among these, the XOS functions, as
defined below, have a particular kinship with subadditive functions.
XOS literally means XOR (taking the maximum) of OR's (taking sums), and this class of valuations is known to be equivalent to the class of \emph{fractionally subadditive} functions \citep{Feige09}.

\begin{definition}
\label{def:XOS}
A function $\val: 2^{[\nitem]} \to \R_+$ is said to be \emph{XOS} if there exists a collection of additive functions
$a_1(\cdot),\ldots ,a_k(\cdot)$ (that is, $a_i(\bundle) \coloneqq \sum\limits_{j\in\bundle}a_i(\{j\})$ for every set
$\bundle\subset[\nitem]$), such that for each $\bundle \subseteq [\nitem]$,
$
\val(\bundle) \coloneqq \max_{1\le i\le k}a_i(\bundle).
$
\end{definition}

One of the characterizations of XOS functions uses the following definition.

\begin{definition}
\label{def:dominance}
A function $f(\cdot)$ is said to be \emph{dominated} by a set function $g(\cdot)$ if
for any subset $\bundle \subseteq [\nitem]$, $f(\bundle) \leq g(\bundle)$.  We say that a vector $\vec{a}=(a_1,\dots,a_\nitem)$ is dominated by a set function $\val(\cdot)$, if as an additive function $a(\cdot)$ is dominated by $\val(\cdot)$.
\end{definition}

It is not too difficult to observe that $\val(\cdot)$ is XOS
if and only if for every set $T\subset[\nitem]$ there is an additive function $a(\cdot)$ dominated by $\val(\cdot)$ such that $a(T)=\val(T)$.

For a general subadditive function $\val(\cdot)$, it can be the case that any additive function $a(\cdot)$ dominated by $\val(\cdot)$ has
$\Omega(\log(\nitem))$ gap from $\val([\nitem])$, i.e. $\Omega(\log(\nitem))a([\nitem])\le \val([\nitem]),$ (See
\citealp{BR11} for such an example) and a logarithmic factor is also an upper bound.
Previous work that attempted to bound the BPoA for subadditive valuations \citep{BR11,HKMN11} provided constant bounds for XOS valuations, then used the logarithmic factor separation between XOS and subadditive valuations to establish a logarithmic upper bound on the BPoA for subadditive valuations.
While it seems plausible to use the close relation between XOS and subadditive valuations, any analysis that follows this trajectory would encounter this inevitable logarithmic gap.
The challenge, therefore, is in developing a new proof technique for subadditive valuations, which does not go through XOS valuations.
This is the approach taken in this work.


\subsection{Overbidding}
\label{sec:nob}
It is well known that in second price auctions, even with only a single item, the price of anarchy can be infinite when bidders are not
restricted in their bids\footnote{A canonical example is two bidders who value the item at $0$ and a
large number~$h$, respectively, but the first bidder bids $h + 1$ and the second bidder bids~$0$.}.  To exclude such
pathological cases, previous literature \citep[e.g.\@][]{CKS08, BR11} has made the following \emph{no-overbidding}
assumption standard\footnote{We note that such no-overbidding assumptions were also made in other contexts
\citep[e.g.\@][]{LB10, PT10}.}:

\begin{definition}
\label{def:snob}
A bidder is \emph{strongly no-overbidding} if his bid $\bid(\cdot)$ is dominated by his valuation $\val(\cdot)$.
\end{definition}

In other words, a bidder is guaranteed to derive non-negative utility, no matter what are the prices in the market.
Thus strong no overbidding is a strong risk-aversion assumption on the buyers.  One may also consider less risk
concerned bidders---in the following we generalize a weaker assumption of no-overbidding introduced by \citet{FKL12}.

\begin{definition}
\label{def:wnob}
Given a price distribution~$\dist$, a bidder is \emph{weakly no-overbidding} if his bid vector~$\bid$ satisfies that
$\Ex[\price \sim \dist]{\val(\win(\bid, \price))} \geq \Ex[\price \sim \dist]{\bid(\win(\bid, \price))}$, where
$\win(\bid, \price)$ denotes the subset of items he wins when he bids~$\bid$ at price~$\price$, i.e., $\win(\bid,
\price) = \{j \in
[\nitem] \mid \bid(j) \geq \price(j) \}$.
\end{definition}

We will bound BPoA under both weakly and strongly no-overbidding assumptions for simultaneous second price auctions.

\section{Bidding Strategies Under Uncertain Prices}
\label{sec:subadditive}

As discussed in \autoref{sec:prelim}, a bidder in a simultaneous auction faces the problem of maximizing his utility
in presence of uncertain prices (which are the largest bids placed by other bidders).
While this maximization problem is intricate, we show in this section particular bidding strategies that result in
utilities comparable with the bidder's value of the whole bundle minus the expected total
prices.  In other words, given a price distribution $\dist$, it is desired to have a bidding strategy $\bid$ such that
\begin{align}
\label{eq:goodbid}
\Ex[\prices \sim \dist]{\val(\bid, \price)} - \bid([\nitem]) \geq \alpha \val([\nitem]) - \Ex[\price \sim
\dist]{\price([\nitem])},
\end{align}
for some constant~$\alpha$.  Such bidding strategies are key ingredients of the BPoA proofs in later sections, and may be of interest on their own.


For fixed prices, achieving \eqref{eq:goodbid} is trivial, even for $\alpha = 1$; 
indeed, given a price vector $\vprice$, by bidding according to 
$\bid=\price$, a bidder obtains $\val(\bid,\price)- \bid([\nitem]) = \val([\nitem]) - \price([\nitem])$.  
The case in which prices are drawn at random is more intricate, and is the subject of the remainder of this section.

\begin{lemma}[\textbf{Bidding against price distributions}]
\label{cl:bidD}
For any distribution~$\dist$ of prices~$\price$ and any subadditive valuation $\val(\cdot)$
there exists a bid $\bid_0$ such that
\begin{align}
\label{eq:bidD}
\Ex[\price\sim\dist]{\val(\bid_0,\price)}-\bid_0([\nitem]) \ge
\frac{1}{2}\val([\nitem])  - \Ex[ \price\sim\dist]{\price([\nitem])}.
\end{align}
\end{lemma}

\begin{proof}
We show a random bidding strategy that guarantees the desired inequality in expectation, and infer the existence of a bid, drawn from the suggested distribution, that achieves the same inequality.
Consider a bid that is drawn according to the exact same distribution as the prices. 
It holds that
\begin{align*}
\Ex[\bid\sim\dist]{\Ex[\price\sim\dist]{\val(\bid,\price)}}  =  \Ex[\price\sim\dist]{\Ex[\bid\sim\dist]{\val(\bid,\price)}} &=
\frac{1}{2} \Ex[\bid\sim\dist]{\Ex[\price\sim\dist]{\val(\bid,\price)+\val(\price,\bid)}}\\
&\ge \frac{1}{2} \Ex[\bid\sim\dist]{\Ex[\price\sim\dist]{\val([\nitem])}} = \frac{1}{2}\val([\nitem]),
\end{align*}
where the inequality follows from subadditivity (which guarantees that $\val(\bid,\price)+\val(\price,\bid) \geq \val([\nitem]$ for every $\price$ and $\bid$). 
Using the last inequality, it follows that
\[
\Ex[\bid\sim\dist]{\Ex[\price\sim\dist]{\val(\bid,\price)}-\bid([\nitem])} \ge
\frac{1}{2}\val([\nitem])-\Ex[\bid\sim\dist]{\bid([\nitem])}=\frac{1}{2}\val([\nitem])-\Ex[\price\sim\dist]{\price([\nitem])}.
\]
Since a bid drawn from $\dists$ satisfies \eqref{eq:bidD} in expectation, there must exist a bid $\bid_0$ satisfying \eqref{eq:bidD}.
  
%
\end{proof}

\subsection{No-Overbidding Strategies Under Uncertain Prices}
\label{sec:nob-price}
As noted in \autoref{sec:nob}, in order to obtain any meaningful bound on BPoA for second price auctions, one needs to assume that
bidders are not overbidding. 
Unfortunately, \autoref{cl:bidD} is not concerned with such requirements.
This problem is addressed in \autoref{cl:bidDsafe}, where it is shown that a strongly no-overbidding strategy analogous to that in
\autoref{cl:bidD} always exists.

Notably, when the no-overbidding requirement is imposed, the existence of a bid satisfying \eqref{eq:goodbid} is already nontrivial
when the prices are fixed.  
The following lemma, rephrased from \cite{BR11}, establishes its existence:


\begin{lemma}[{\bf Lemma 3.3 in \citealp{BR11}}]
For a given price vector $\price$ and any subadditive valuation $\val(\cdot)$
there exists a bid $\bid$ dominated by $v(\cdot)$ such that
\begin{align*}
\val(\bid, \price)- \bid([\nitem]) \ge \val([\nitem]) - \price([\nitem]).
\end{align*}
\end{lemma}

We must now analyze the case when prices are drawn randomly.

\begin{lemma}[\bf{No Overbidding Against Price Distributions}]
\label{cl:bidDsafe}
For any distribution~$\dist$ of prices~$\price$ and any subadditive valuation $\val(\cdot)$
there exists a bid $\bid_0$ dominated by $\val(\cdot)$ such that
\begin{equation}
\label{eq:bidDsafe}
\Ex[\price\sim\dist]{\val(\bid_0,\price)} - \bid_0([\nitem]) \ge
\frac{1}{2}\val([\nitem]) - \Ex[\price\sim\dist]{\price([\nitem])}.
\end{equation}
\end{lemma}
\begin{proof}

Let $\qprice$ be any price vector in the support of the distribution $\dist$.
Let $T\subseteq[\nitem]$ be a maximal set such that $\val(T)\le \qprice(T)$.
We consider a truncated price vector $\qpricew$, which is $0$ on the coordinates corresponding to $T$ and coincides with
$\qprice$ on the coordinates corresponding to $[\nitem]\setminus T$.

We first observe that {\bf $\qpricew$ is dominated by $\val(\cdot).$}
Indeed, for any set $R \subset[\nitem]\setminus T$ it holds that $\val(R) > \qprice(R)$,
since otherwise
\[
\val(R\cup T)\le\val(R)+\val(T) \le \qprice(R)+\qprice(T)=\qprice(R\cup T),
\]
in contradiction to the fact that $T$ is a maximal set satisfying $\val(T)\le \qprice(T)$.

We next establish that for any bid $\bid$, it holds that
\begin{align}
\label{eq:priceMod}
\val(\bid,\qprice)+\qprice([\nitem])\ge \val(\bid,\qpricew)+\qpricew([\nitem]).
\end{align}

Indeed, we have $\win(\bid,\qpricew)\subseteq \win(\bid,\qprice) \cup T$. Therefore,
$\val(\bid,\qpricew)\le \val(\bid,\qprice) + v(T)$ due to subadditivity of $\val(\cdot)$.
Now \eqref{eq:priceMod} follows by observing that $\qprice([\nitem])-\qpricew([\nitem])=\qprice(T)\ge v(T).$

We next define the distribution $\distw \coloneqq \left\{\qpricew\mid \qprice\sim\dist \right\}$ which
consist of truncated prices drawn from $\dist$. Equation \eqref{eq:priceMod} now extends for any bid $\bid$ to
\begin{align}
\label{eq:DMod}
\Ex[\price\sim\dist]{\val(\bid,\price)+\price([\nitem])}\ge \Ex[\pricew\sim\distw]{\val(\bid,\pricew)
+\pricew([\nitem])}.
\end{align}

Recall that each $\qpricew\sim\distw$ is dominated by $\val(\cdot)$, therefore, bidding any $\bid$ drawn from $\distw$
satisfies the strongly no overbidding requirement.  
Furthermore, by applying \eqref{eq:DMod} to each $\bid\sim\distw$ we get
\begin{align*}
\Ex[\bid\sim\distw]{\Ex[\price\sim\dist]{\val(\bid,\price)+\price([\nitem])}}  & \ge
\Ex[\bid\sim\distw]{\Ex[\pricew\sim\distw]{\val(\bid,\pricew)+\pricew([\nitem])}}\\
&= \Ex[\bid\sim\distw]{\Ex[\pricew\sim\distw]{\val(\bid,\pricew)}}+\Ex[\bid\sim\distw]{\bid([\nitem])}\\
&\ge \frac{1}{2}\val([\nitem])+\Ex[\bid\sim\distw]{\bid([\nitem])},
\end{align*}
where the last inequality follows in a manner similar to the proof of \autoref{cl:bidD}. 
The assertion of the lemma follows
\end{proof}

\section{Price of Anarchy for First Price Auctions}
\label{sec:poa-fpa}
In this section we apply the bidding strategy from \autoref{cl:bidD} to bound the Bayesian price of anarchy of simultaneous first-price auctions.

%
%
%
%

\begin{theorem}
\label{thm:poa-fpa}
In a simultaneous first-price auction with subadditive bidders, the Bayesian price of anarchy is at most $2$.
\end{theorem}

\begin{proof}
Fix type distributions $\valdists$ and let $\strats$ be a BNE for $\valdists$.  Choose some agent $i$ and an arbitrary subadditive valuation $\vali$. Fix an arbitrary $\valsmi^*$, and let $\vals^* = (\vali, \valsmi^*)$.  Recall that $(\opt_1^{\vals^*}, \dotsc, \opt_n^{\vals^*})$ is the welfare-optimal allocation for $\vals^*$.

Recall that each bid profile $\bidsmi$ defines for bidder~$i$ a price vector $\projbid_i(\bidsmi)$.  Let $\vprice$ be
equal to $\projbid_i(\bidsmi)$ on $\opt_i^{\vals^*}$ and $0$ elsewhere.  
Let $\dists$ be the distribution over these price vectors
$\vprice = \vprice(\bidsmi)$, where $\bids\sim\strats(\vals^*)$. That is, $\dists$ is precisely the distribution over the maximum
bids on the items in $\opt_i^{\vals^*}$, 
excluding the bid of player~$i$.  By \autoref{cl:bidD} (and replacing $[\nitem]$ there by $\opt_i^{\val^*}$), there exists a bid vector $\bidi'$ 
over the objects in $\opt_i^{\vals^*}$ such that, thinking now of $\price$ as an additive function,
\begin{align}
\label{eq:poa-fpa1}
  \Ex[ \price\sim\dist]{\vali(\bidi',\price)}-\bidi'(\opt_i^{\vals^*})
     \ge
  \frac{1}{2}\vali(\opt_i^{\vals^*})  - \Ex[ \price\sim\dist]{\price(\opt_i^{\vals^*})}.
\end{align}
Since $\strats$ forms a BNE, we have that
\begin{align*}
\Exlong[\substack{\valsmi,\\ \bids\sim\strats(\vals)}]{\util_i(\bids)} & \geq \Exlong[\substack{\valsmi,\\ \bids\sim\strats(\vals)}]{ \util_i(\bidi', \bidsmi) }
= \Exlong[\substack{\valsmi,\\ \bids\sim\strats(\vals)}]{\vali(\bidi', \projbid_i(\bidsmi))} - \Exlong[\substack{\valsmi,\\ \bids\sim\strats(\vals)}]{\bidi'(\wini(\bidi',\bidsmi))} \\
& \geq \Ex[\price\sim\dist]{\vali(\bidi',\price)} - \bidi'(\opt_i^{\vals^*}),
\end{align*}
where the last inequality follows from the definition of $\dists$ and the fact that $\wini(\bidi',\bidsmi) \subseteq \opt_i^{\vals^*}$ for all $\bidsmi$.  Applying \eqref{eq:poa-fpa1} and the definition of $\price\sim\dists$, we conclude that
\begin{equation}
\label{eq:poa-fpa2}
  \Exlong[\substack{\valsmi,\\ \bids\sim\strats(\vals)}]{\util_i(\bids)}  \geq
  \frac{1}{2}\vali(\opt_i^{\vals^*})  - \Exlong[\substack{\valsmi,\\ \bidsmi\sim\stratsmi(\valsmi)}]{\sum_{j \in \opt_i^{\vals^*}}\max_{k \neq i}\bid_k(j)}.
\end{equation}
Taking the sum over all $i$ and expectations over all $\vali\sim\valdist_i$ and $\valsmi^*\sim\valdist_{-i}$, we conclude that
\begin{equation}
\label{eq:poa-fpa3}
  \sum_i \Exlong[\substack{\vals, \valsmi^*,\\ \bids\sim\strats(\vals)}]{\util_i(\bids)}
    \geq
  \frac{1}{2}\sum_i \Exlong[\vali, \valsmi^*]{\vali(\opt_i^{\vals^*})}  - \sum_i\Exlong[\substack{\vals, \valsmi^*,\\ \bidsmi\sim\stratsmi(\valsmi)}]{\sum_{j \in \opt_i^{\vals^*}}\max_{k \neq i}\bid_k(j)}.
\end{equation}
Let us consider each of the three terms of \eqref{eq:poa-fpa3} in turn.  The LHS is equal to $\Ex[\vals,\bids\sim\strats(\vals)]{\sum_i \util_i(\bids)}$, as $\valsmi^*$ does not appear inside the expectation.  The first term on the RHS is equal to $\frac{1}{2}\Ex[\vals]{\sum_i \vali(\opt_i^{\vals})}$, by relabeling
$\valsmi^*$ by $\valsmi$.
For the final term on the RHS of \eqref{eq:poa-fpa3}, we note that
\begin{align*}
  \sum_i\Exlong[\substack{\vals, \valsmi^*,\\ \bidsmi\sim\stratsmi(\valsmi)}]{\sum_{j \in \opt_i^{\vals^*}}\max_{k \neq i}\bid_k(j)}
& \leq
  \sum_i\Exlong[\substack{\vals, \valsmi^*, \hvali,\\ \bids\sim\strats(\hvali,\valsmi)}]{\sum_{j \in \opt_i^{\vals^*}}\max_{k}\bid_k(j)} \\
& =
  \Exlong[\vals, \bids\sim\strats(\vals)]{\sum_{j}\max_{k}\bid_k(j)},
\end{align*}
where the first inequality is due to taking a maximum over a larger set, and the last equality follows since $\opt_i^{\vals^*}$ form a partition of $[m]$ (and by relabeling).  We note a subtlety: in the first line we select bid vector $\bids$ with respect to $(\hvali, \valsmi)$, rather than $(\vali, \valsmi)$, so that $\bids$ is independent of the partition $(\opt_1^{\vals^*}, \dotsc, \opt_n^{\vals^*})$.  Applying these simplifications to the terms of \eqref{eq:poa-fpa3}, we conclude that
\begin{equation}
\label{eq:poa-fpa4}
\Ex[\vals, \bids\sim\strats(\vals)]{\sum_i \util_i(\bids)} \geq \frac{1}{2}\Ex[\vals]{\sum_i \vali(\opt_i^{\vals})} -
\Ex[\vals, \bids\sim\strats(\vals)]{\sum_{j}\max_{k}\bid_k(j)}.
\end{equation}
Since we are in a first-price auction, we have that $\Ex[\vals,\bids\sim\strats(\vals)]{\sum_i \util_i(\bids)} = \Ex[\vals,\bids\sim\strats(\vals)]{\sum_i \vali(\wini(\bids))} - \Ex[\vals, \bids\sim\strats(\vals)]{\sum_{j}\max_{k}\bid_k(j)}$.  Equation \eqref{eq:poa-fpa4} therefore implies
\[ \Ex[\vals,\bids\sim\strats(\vals)]{\sum_i \vali(\wini(\bids))} \geq \frac{1}{2}\Ex[\vals]{\sum_i \vali(\opt_i^{\vals})} \]
which yields the desired result.
\end{proof}

\noindent {\bf Remark:} In \autoref{sec:lb-correlated}, we show that the upper bound does not carry over to the case where the bidders' valuations are correlated. 
In particular, a polynomial lower bound of $\Omega(n^{1/6})$ is given on the Bayesian price of anarchy for this case. 
The construction is based heavily upon a lower bound due to \cite{BR11} for second-price auctions. 

\section{Price of Anarchy for Second Price Auctions}
\label{sec:poa-spa}

We now turn to the case of simultaneous second-price auctions.  We show that the Bayesian price of anarchy of such an auction
is always at most $4$ for subadditive bidders, assuming that bidders select strategies that satisfy either the
strong or weak no-overbidding assumption.


\begin{theorem}
\label{thm:snob-poa}
In simultaneous second price auctions where bidders have subadditive valuations independently drawn and each of them
is strongly or weakly no-overbidding, the Bayesian price of anarchy is at most~$4$.
\end{theorem}
\begin{proof}
Fix type distributions $\valdists$ and let $\strats$ be a BNE for $\valdists$.  We can then derive inequality
\eqref{eq:poa-fpa4} in precisely the same way as in the proof of \autoref{thm:poa-fpa} (using now \autoref{cl:bidDsafe}
instead of \autoref{cl:bidD}); we then have that
\begin{equation}
\label{eq:snob-poa1}
\Ex[\vals,\bids\sim\strats(\vals)]{\sum_i \util_i(\bids)} \geq \frac{1}{2}\Ex[\vals]{\sum_i \vali(\opt_i^{\vals})} - \Ex[\vals, \bids\sim\strats(\vals)]{\sum_{j}\max_{k}\bid_k(j)}.
\end{equation}
Note that $\Ex[\vals,\bids\sim\strats(\vals)]{\sum_i \vali(\win_i(\bids))} \geq \Ex[\vals,\bids\sim\strats(\vals)]{\sum_i \util_i(\bids)}$.  Also, since each agent $i$
is assumed to be strongly or weakly no overbidding,
\begin{align*}
\Ex[\vals, \bids\sim\strats(\vals)]{\sum_{j}\max_{k}\bid_k(j)}
= \Ex[\vals, \bids\sim\strats(\vals)]{\sum_i \sum_{j \in \win_i(\bids)}\bidi(j)}
\leq \Ex[\vals, \bids\sim\strats(\vals)]{\sum_i \vali( \win_i(\bids))}.
\end{align*}
Equation \eqref{eq:snob-poa1} therefore implies
\[
  \Ex[\vals,\bids\sim\strats(\vals)]{\sum_i \vali(\win_i(\bids))}
\geq
  \frac{1}{2}\Ex[\vals]{\sum_i \vali(\opt_i^{\vals})} - \Ex[\vals, \bids\sim\strats(\vals)]{\sum_i \vali( \win_i(\bids))}
\]
as required.
\end{proof}

\citeauthor{BR11} showed that the Bayesian price of anarchy of second price auctions can be strictly worse than the pure
price of anarchy when bidders are strongly no overbidding.  In the following we give an example showing that such a gap
exists also when bidders are weakly no overbidding.
We note that this gap is not implied by the example given by \citeauthor{BR11} since the strategy profile in their example is not a BNE under the weaker no overbidding notion.

\begin{example}[\textbf{Bayesian price of anarchy can be strictly larger than~$2$ when bidders are weakly no overbidding and have subadditive valuations}]
\label{ex:wnob-poa}
Consider an instance with 2 bidders and 6 items, where the set of items is divided into two sets, of 3 items each, denoted $S_1$ and $S_2$.
Throughout, we shall present the example with parameters $a$ and $b$ for ease of presentation.
The lower bound is obtained by substituting $a=0.06$ and $b=0.85$.
In what follows, we describe the valuation function of bidder 1; bidder 2's valuation is symmetric w.r.t. the sets $S_1$ and $S_2$.
Bidder 1's valuation over the items in $S_1$ is additive with respective values (over the 3 items) of $(a,a,b), (b,a,a)$ or $(a,b,a)$, each with probability $1/3$.
Bidder 1's valuation over the items in $S_2$ is 2 if she gets all three items, and 1 for any non-empty strict subset of $S_2$.
Bidder 1's valuation for an arbitrary subset $T$ the maximum of her value for $T \cap S_1$ and her value for $T \cap S_2$.
One can verify that this is indeed a subadditive valuation function.

We claim that the profile in which each bidder $i$ bids her true (additive) valuation on $S_i$ and $0$ on all other items 
is a Bayesian equilibrium with weakly no overbidding bidders for the specified parameter values.
The full proof is deferred to the appendix, where it is shown that the only beneficial deviations break the weakly no-overbidding assumption.
Under this bidding profile, each bidder derives a utility of $2a+b$, amounting to a social welfare of $2(2a+b) = 1.94$.
In contrast, if bidder 1 is allocated $S_2$ and bidder 2 is allocated $S_1$, then each bidder derives a utility of $2$, amounting to a social welfare of $4$.
Consequently, the Bayesian price of anarchy is $4/1.94 > 2.061$.
\end{example}

\bibliographystyle{apalike}
\bibliography{bibs}

\begin{thebibliography}{}

\bibitem[Alon et~al., 2010]{Alon10}
Alon, N., Emek, Y., Feldman, M., and Tennenholtz, M. (2010).
\newblock Bayesian ignorance.
\newblock In {\em Proceedings of the 29th ACM SIGACT-SIGOPS symposium on
  Principles of distributed computing}, PODC '10, pages 384--391, New York, NY,
  USA. ACM.

\bibitem[Bhawalkar and Roughgarden, 2011]{BR11}
Bhawalkar, K. and Roughgarden, T. (2011).
\newblock Welfare guarantees for combinatorial auctions with item bidding.
\newblock In {\em SODA}, pages 700--709.

\bibitem[Bikhchandani, 1999]{Bikhchandani99}
Bikhchandani, S. (1999).
\newblock Auctions of heterogeneous objects.
\newblock {\em Games and Economic Behavior}, 26(2):193--220.

\bibitem[Caragiannis et~al., 2011]{CKKK11}
Caragiannis, I., Kanellopoulos, P., Kaklamanis, C., and Kyropoulou, M. (2011).
\newblock On the efficiency of equilibria in generalized second price auctions.
\newblock In {\em EC'11}.

\bibitem[Christodoulou et~al., 2008]{CKS08}
Christodoulou, G., Kov{\'a}cs, A., and Schapira, M. (2008).
\newblock Bayesian combinatorial auctions.
\newblock In {\em 35th International Colloquium on Automata, Languages and
  Programming, 35th International Colloquium}, pages 820--832.

\bibitem[Edelman et~al., 2005]{Edelman05}
Edelman, B., Ostrovsky, M., Schwarz, M., Fudenberg, T.~D., Kaplow, L., Lee, R.,
  Milgrom, P., Niederle, M., and Pakes, A. (2005).
\newblock Internet advertising and the generalized second price auction:
  Selling billions of dollars worth of keywords.
\newblock {\em American Economic Review}, 97.

\bibitem[Feige, 2009]{Feige09}
Feige, U. (2009).
\newblock On maximizing welfare when utility functions are subadditive.
\newblock {\em SIAM J. Comput.}, 39(1):122--142.

\bibitem[Fu et~al., 2012]{FKL12}
Fu, H., Kleinberg, R., and Lavi, R. (2012).
\newblock Conditional equilibrium outcomes via ascending price processes with
  applications to combinatorial auctions with item bidding.
\newblock In {\em ACM Conference on Electronic Commerce}, page 586.

\bibitem[Hassidim et~al., 2011]{HKMN11}
Hassidim, A., Kaplan, H., Mansour, Y., and Nisan, N. (2011).
\newblock Non-price equilibria in markets of discrete goods.
\newblock In {\em ACM Conference on Electronic Commerce}, pages 295--296.

\bibitem[Jackson et~al., 2002]{JSSZ02}
Jackson, M.~O., Simon, L.~K., Swinkels, J.~M., and Zame, W.~R. (2002).
\newblock Communication and equilibrium in discontinuous games of incomplete
  information.
\newblock {\em Econometrica}, 70(5):1711--1740.

\bibitem[Lehmann et~al., 2006]{LLN06}
Lehmann, B., Lehmann, D.~J., and Nisan, N. (2006).
\newblock Combinatorial auctions with decreasing marginal utilities.
\newblock {\em Games and Economic Behavior}, 55(2):270--296.

\bibitem[Lucier and Borodin, 2010]{LB10}
Lucier, B. and Borodin, A. (2010).
\newblock Price of anarchy for greedy auctions.
\newblock In {\em 21st ACM-SIAM Symposium on Discrete Algorithms}, pages
  537--553.

\bibitem[Lucier and {Paes Leme}, 2011]{LP11}
Lucier, B. and {Paes Leme}, R. (2011).
\newblock Gsp auctions with correlated types.
\newblock In {\em EC'11}.

\bibitem[Mamer and Bikhchandani, 1997]{BM97}
Mamer, J. and Bikhchandani, S. (1997).
\newblock {\em Journal of Economic Theory}, 74:385--413.

\bibitem[Milgrom, 1998]{Milgrom98}
Milgrom, P. (1998).
\newblock Putting auction theory to work: The simultaneous ascending auction.
\newblock {\em Journal of Political Economy}, 108:245--272.

\bibitem[Nisan et~al., 2007]{NRTV07}
Nisan, N., Roughgarden, T., Tardos, E., and Vazirani, V.~V. (2007).
\newblock {\em Algorithmic Game Theory}.
\newblock Cambridge University Press, New York, NY, USA.

\bibitem[Paes~Leme and Tardos, 2010]{PT10}
Paes~Leme, R. and Tardos, {\'E}. (2010).
\newblock Pure and bayes-nash price of anarchy for generalized second price
  auction.
\newblock In {\em FOCS}, pages 735--744.

\bibitem[Roughgarden, 2012]{R12}
Roughgarden, T. (2012).
\newblock The price of anarchy in games of incomplete information.
\newblock In {\em ACM Conference on Electronic Commerce}, pages 862--879.

\bibitem[Roughgarden and Tardos, 2007]{RT07}
Roughgarden, T. and Tardos, E. (2007).
\newblock Introduction to the inefficiency of equilibria.
\newblock In Nisan, N., Roughgarden, T., Tardos, E., and Vazirani, V.~V.,
  editors, {\em Algorithmic Game Theory}. Cambridge University Press, New York,
  NY, USA.

\bibitem[Simon and Zame, 1990]{SZ90}
Simon, L.~K. and Zame, W.~R. (1990).
\newblock Discontinuous games and endogenous sharing rules.
\newblock {\em Econometrica}, 58(4):861--72.

\bibitem[Syrgkanis, 2012]{S12}
Syrgkanis, V. (2012).
\newblock Bayesian games and the smoothness framework.
\newblock {\em CoRR}, abs/1203.5155.

\bibitem[Varian, 2007]{Varian07}
Varian, H.~R. (2007).
\newblock Position auctions.
\newblock {\em International Journal of Industrial Organization},
  25(6):1163--1178.

\end{thebibliography}

\appendix

\section{A Proof of Example~\ref{ex:wnob-poa}}
\label{sec:lb-poa-spa}
In this section we prove that the strategy profile in \autoref{ex:wnob-poa} is a Bayesian equilibrium with weakly no overbidding bidders.
To establish this, we need to show that every beneficial deviation breaks the weakly no overbidding assumption.
Notably, since weak no-overbidding is required for every bid in the support, it is sufficient to consider only pure deviations. 
By symmetry, it suffices to consider only deviations by bidder $1$.  Finally, it suffices to consider only deviations in which bidder 1 bids either $0$, $a$, or $b$ on each item in $S_2$ and $0$ on all items in $S_1$; this is because bidder $1$ obtains value from either $S_1$ or $S_2$ but not both, and without deviating bidder $1$ obtains all of $S_1$ at no cost.

The following table includes all the possible bids (in the rows), and their respective expected values, expected payments and expected bids (in the columns).
For clarity of presentation, we present the expressions in parametric forms, and write the corresponding values for $a=0.06$ and $b=0.85$ in brackets.

\begin{center}
\begin{tabular}{|c|c|c|c|}
\hline
Deviation & $\Ex[\price \sim \dist]{\val(\win(\bid, \price))}$ & $\Ex[\price \sim \dist]{\price(\win(\bid, \price))}$ & $\Ex[\price \sim \dist]{\bid(\win(\bid, \price))}$\\
\hline
$(a,a,a)$ & $1$ & $2a [$0.12$]$ & $2a$ [$0.12$]\\
\hline
$(a,a,b)$ & $\third \cdot 2 + \twothirds \cdot 1 = \fourthirds$ & $\third(2a+b)+\twothirds(2a)=2a+\third b$ [$0.4033..$] & $\third(2a+b)+\twothirds(a+b)=\fourthirds a+b$ [$0.93$]\\
\hline
$(a,b,b)$ & $\twothirds \cdot 2 + \third \cdot 1 = \frac{5}{3}$ & $\twothirds(2a+b)+\third(2a)=2a+\twothirds b$ [$0.6866..$]& $\twothirds(a+2b)+\third(2b)=\twothirds a+ 2b$ [$1.74$]\\
\hline
$(b,b,b)$ & $2$ & $2a+b$ [$0.97$]& $3b$ [$2.55$]\\
\hline
$(a,0,0)$ & $\twothirds \cdot 1 + \third \cdot 0 = \twothirds$ & $\twothirds a$ [$0.04$]& $\twothirds a$ [$0.04$]\\
\hline
$(a,a,0)$ & $1$ & $\third(2a)+\twothirds a= \fourthirds a$ [$0.08$]& $\third(2a)+\twothirds a= \fourthirds a$ [$0.08$]\\
\hline
$(b,0,0)$ & $1$ & $\twothirds a+\third b$ [0.2833..]& $b$ [0.85]\\
\hline
\end{tabular}
\end{center}

It is evident from the table that for deviations $(a,b,b)$ and $(b,b,b)$, $\Ex{\val(\win(\bid, \price))} < \Ex{\bid(\win(\bid, \price))}$, and therefore they do not satisfy weakly no overbidding. 
For each of the remaining deviations, the obtained expected utility (which equals $\Ex{\val(\win(\bid, \price))} - \Ex{\price(\win(\bid, \price))}$) is smaller than the current expected utility (which equals $2a+b = 0.97$).
We conclude that the strategy profile in the example is a Bayesian equilibrium with weakly no overbidding bidders, as required.

\section{A Lower Bound for Correlated Valuations}
\label{sec:lb-correlated}
In this section we give a polynomial lower bound, $\Omega(n^{1/6})$, on the Bayesian price of anarchy for first-price auctions with subadditive valuations, when the valuation distributions are correlated among the bidders.  In fact, our example will hold even when all valuations are unit demand.  The construction is based heavily upon a lower bound due to \cite{BR11} for second-price auctions.

\begin{example}[ \textbf{High price of anarchy for correlated valuations and weakly no-overbidding players} ]
\label{ex:lb-correlated}
There are $n+(n+1)\sqrt{n}$ items and $3n$ players.  Players occur in triples.  Each triple contains one player of type $I$ and two players of type $II$.  A valuation from the correlated distribution $\dists$ is drawn as follows.  First, a set $T$ of $\sqrt{n}$ items are selected at random; we will refer to these items as the common pool.  Next, $n$ of the remaining items are selected at random and labelled $a_1, \dotsc, a_n$; we refer to these as the reserve items.  Finally, the remaining $n\sqrt{n}$ items are partitioned into sets $S_1, \dotsc, S_n$, each of size $\sqrt{n}$; we refer to these as the mock pools.  Reserve item $a_i$ and mock pool $S_i$ are matched with the $i$th triple of players.

Given the labelling of the items, the player valuations are as follows.  There are two possibilities for the valuation profile; an atypical case that occurs with probability $p = \frac{1}{n^{1/6}}$, and a typical case that occurs with the remaining probability $1-p$.  In the typical case, each player of type $II$ has value $n^{-1/6}$ for the corresponding reserve item $a_i$, and each player of type $I$ has value $1$ for any non-empty subset of the common pool plus the corresponding reserve item, $T \cup \{a_i\}$.  In the atypical case, each player of type $II$ has the zero valuation, and each player of type $I$, say from triple $i$, has value $1$ for any non-empty subset of the corresponding mock pool plus reserve item, $S_i \cup \{a_i\}$.  

First note that we can assume in a Bayes-Nash equilibrium that each player of type $II$ always bids $n^{-1/6}$ on his reserve item, in the typical case.  Bidding more than $n^{-1/6}$ leads to negative utility if he wins, and bidding less than $n^{-1/6}$ allows the other type $II$ bidder in the triple to obtain positive utility by winning the item with a bid less than $n^{-1/6}$.  Thus both agents of type $II$ in a triple will bid $n^{-1/6}$, causing both to have utility $0$.  In the atypical case, each type $II$ bidder trivially bids $0$ on all items.

A player of type $I$, when bidding in equilibrium, cannot distinguish between the typical and atypical cases; he always sees a set of $\sqrt{n} + 1$ items for which he has value, and each item is equally likely to be the reserve item.  Note that it has to bid at least $n^{-1/6}$ on the reserve item in order to win it in the typical case.  Suppose that the player bids at least $n^{-1/6}$ on some number $k$ of the $\sqrt{n}+1$ items.  Then if the valuation profile is atypical the player will win all $k$ items, and pay at least $k \cdot n^{-1/6}$.  The expected payment of the player is therefore at least $pkn^{-1/6} = kn^{-1/3}$.  If $k > n^{1/3}$ then the expected payment of the player is greater than $1$, and hence his expected utility is negative, contradicting the assumption of Bayes-Nash equilibrium.  We therefore conclude that $k \leq n^{1/3}$.  Each player of type $I$ will therefore win its reserve item in the typical case with probability at most $k / (\sqrt{n} + 1) < n^{-1/6}$.  

We conclude that the social welfare of any Bayes-Nash equilibrium $\strats$ satisfies
\[ \Ex[\vals,\bids\sim\strats(\vals)]{\sum_i \vali(\wini(\bids))} \leq p n + (1-p)(n \cdot n^{-1/6} + n \cdot n^{-1/6} \cdot 1 + \sqrt{n} \cdot 1) = O(n^{5/6}) \]
where the expression for the typical case includes bounds on the value obtained by the type $II$ bidders, the value of the type $I$ bidders who win reserve items, and the value of the type $I$ bidders who win items from the common pool, respectively.  Since the optimal social welfare is at least $n$ in each case, the price of anarchy is at least $\Omega(n^{1/6})$.
\end{example}

\section{No Overbidding: a Discussion}
\label{sec:no-overbidding-discussion}

In our analysis of the BPoA of second-price auctions we have adopted either the strong version or the weak version of the no-overbidding assumption.
A few conceptual remarks are in order.

We can think of no-overbidding assumptions as representing a form of risk aversion.  The strong no-overbidding assumption guarantees to the bidder a non-negative utility, independent of the behavior of the other players; i.e., even if the other players behave in an arbitrary way.
The weak no-overbidding assumption, in contrast, guarantees to the bidder a non-negative utility only if the other bidders behave ``as expected".
However, when the other bidders behave as expected, the bidder is guaranteed a non-negative utility even if the auction changes, ex-post, from a second-price auction to a first-price auction.

Let us give an example to illustrate the difference between the two assumptions.  Consider an instance of a simultaneous second-price auction with two bidders and two items, say $\{a,b\}$.  The first bidder is unit-demand; with probability $1$ his valuation is such that he has value $1$ for any non-empty subset of the items.  The second bidder's valuation is additive, and distributed as follows: with probability $1/2$ she values $a$ for $0.9$ and $b$ for $1.1$, and with the remaining probability $1/2$ she values $a$ for $1.1$ and $b$ for $0.9$.  In this instance, since the second bidder's valuation is additive it is a dominant strategy for her to bid her true value on each item.  The best response for the first bidder is then to bid between $0.9$ and $1$ on each item: this guarantees that he wins one of the items and pays $0.9$.  This profile of strategies then forms a BNE for this instance.  This bidding strategy of player $1$ does not satisfy the strong no-overbidding assumption: it requires that he indicate a value of at least $1.8$ for the set $\{a,b\}$, which is larger than his true value $1$.  However, it does satisfy the weak no-overbidding assumption given the behavior of bidder $2$, since bidder $1$ expects to win only one item (of value $1$) with a bid of $0.9$.

The above example illustrates a situation in which the best response of a player is permitted by weak no-overbidding but excluded by strong no-overbidding.  There also exist cases in which a best response is also excluded by the weak no-overbidding assumption.  \autoref{ex:wnob-poa} is one such case: the players can improve their utilities, but only by applying strategies that violate weak no-overbidding.  A direction for future research would be to determine whether there is a weaker restriction on strategies that never excludes best-responses, but yet still guarantees a constant price of anarchy bound.

The use of no-overbidding assumptions in Vickrey auctions and GSP auctions \citep{PT10,LP11} was justified by the fact that overbidding is weakly dominated: any overbidding strategy can be converted to a no-overbidding strategy that performs at least as well, regardless of the behavior of the other agents.  For the case of simultaneous item auctions,
our no-overbidding assumption cannot be relaxed to the assumption that bidders avoid such {\em dominated strategies}.
In particular, there exists an instance of a second-price auction with a Bayesian equilibrium in which all bidders play undominated strategies, and the Bayesian price of anarchy is $\Omega(n)$.
For example, consider an instance with $n$ unit-demand bidders and $n$ items, where every bidder $i=1, \ldots, n-1$ values each of item $i$ and item $n$ at $1-\epsilon$ (for some $\epsilon>0$), and bidder $n$ values all items $1, \ldots, n-1$ at $1$ (and has no value for item $n$).
One can easily verify that, for bidder $n$, to bid $1$ on all the first $n-1$ items is an undominated strategy (while it obviously breaks the strong no overbidding requirement).
Consider the strategy profile in which bidder $n$ bids according to this strategy, and each of bidders $i=1, \ldots, n-1$ bids $0$ on item $i$ and $1-\epsilon$ on item $n$.
This is a Bayesian equilibrium in undominated strategies in a second-price auction, which gives social welfare $2-\epsilon$, compared to the optimal social welfare, which is roughly $n$.

\section{Existence of Equilibria}
\label{app:existence}

The simultaneous auction games we consider have continuous type spaces (i.e.\ valuations) and continuous (pure) strategy spaces (i.e.\ potential bids).  In general, equilibria may not exist in such infinite games, even when the strategy space is compact.  As a toy example, consider a game in which each bidder declares a value from $[0,1]$, and whoever declares the largest value strictly less than $1$ wins; such a game does not admit any mixed equilibria.  The existence of equilibria in infinite games is an involved topic, a full discussion of which falls outside the scope of this paper; we hope only to give a brief discussion of relevant issues and results.  

Consider first a variant of our auction game in which agent types and bids are discretized and bounded.  That is, suppose that all values lie in $[0,1]$, and moreover that there is some $\epsilon > 0$ such that for each agent $i$ and every set of items $S$, $\vali(S)$ can be expressed as $\epsilon \times k_i(S)$ for some integer $k_i(S) \geq 0$.  Furthermore, each agent is restricted to placing bids from $[0,1]$, each of which must be a multiple of $\epsilon$.  In this restricted game, a Bayes-Nash equilibrium always exists.  To see this, note that the set of (pure) strategies is finite: it is the set of all functions mapping agent types (a finite set) to bid vectors (also finite).  We can interpret the (Bayesian) game of incomplete information as the following normal-form game: each agent selects a bidding function ex ante, and the payoffs correspond to the expected payoffs in the Bayesian game under the commonly known distribution of types.  Since the strategy space is finite, Nash's result implies the existence of a mixed Nash equilibrium of this normal-form game, which corresponds precisely to a Bayes-Nash equilibrium of the original game.

Let us turn now to the more general setting of continuous agent valuations and bids, say constrained to lie in $[0,1]$.  We can, of course, approximate the continuous setting via discretizations to an $\epsilon$-grid with arbitrarily small choice of $\epsilon$.  It is tempting to take the limit $\epsilon \to 0$, but the existence of an equilibrium in each approximate games does not necessarily imply the existence of an equilibrium for the limit case, even in settings of complete information.  Consider, for example, the sale of a single object via first-price auction between two bidders, where the first bidder has value $1/2$ and the second bidder has value $1$, in which ties are broken in favor of bidder 1.  The natural equilibrium in this case is for bidder 1 to bid $1/2$ and bidder 2 to bid slightly higher, but no fixed bid of bidder 2 is optimal: any bid of the form $1/2 + \delta$ is strictly worse than $1/2 + \delta/2$ for any $\delta > 0$.  Note that Nash's theorem does not imply the existence of equilibria in this case because utilities are discontinuous in the strategy space: the utility of a player bidding on a single item is discontinuous at the bid value of the highest-bidding competitor.

The issue in the above example lies in the choice of tie-breaking rule.  If ties were broken in favor of bidder 2, there would indeed exist an equilibrium in which each bidder bids $1/2$.  It turns out that this is not coincidental: for the specific case of mixed Nash equilibria in games of complete information, a result due to \cite{SZ90} implies that non-existence of equilibrium is always due to the choice of tie-breaking rule.  Applied in the context of simultaneous item auctions, their result implies that for any profile of agent types, there exists a tie-breaking rule (i.e.\ manner of distributing items for which multiple players declare the same bid) such that a mixed Nash equilibrium exists.  Importantly, the tie-breaking rule used may depend on the types of the agents.

For settings of incomplete information the situation is less clear.  The work of \cite{SZ90} has been extended to Bayes-Nash equilibria under some restrictions on agent utilities, such as in \cite{JSSZ02}.  They demonstrate that in certain auction settings, one can incentivize agents to reveal their types for the purpose of implementing a tie-breaking rule that guarantees the existence of an equilibrium.  However, their approach relies crucially on bidder utilities being affine functions of the auction outcome, which is not necessarily the case for combinatorial auctions with non-additive agent valuations.

To the best of our understanding, for the particular case of simultaneous item auctions for agents with subadditive valuations, prior work does not imply a manner of selecting tie-breaking rules so that BNE always exist.  We therefore view our results as pertaining most directly to discrete approximations of the auction with continuous agent types, and leave for future work the task of establishing (or disproving) BNE existence in general.

\end{document}